\documentclass{article}[12pt]
\usepackage{graphicx,epsfig,color}
\usepackage{enumerate}
\usepackage{lineno,hyperref}
\usepackage{algorithm}
\usepackage[noend]{algpseudocode}
\usepackage{amsfonts}
\usepackage{amsmath}
\newcounter{savenumi}

\newtheorem{theoremfoo}{Theorem}
\newenvironment{theorem}{\pagebreak[1]\begin{theoremfoo}}{\end{theoremfoo}}

\newtheorem{propositionfoo}[theoremfoo]{Proposition}

\newtheorem{lemmafoo}[theoremfoo]{Lemma}
\newenvironment{lemma}{\pagebreak[1]\begin{lemmafoo}}{\end{lemmafoo}}
\newtheorem{conjecturefoo}[theoremfoo]{Conjecture}

\newtheorem{corollaryfoo}[theoremfoo]{Corollary}
\newenvironment{corollary}{\pagebreak[1]\begin{corollaryfoo}}{\end{corollaryfoo}}
\newtheorem{exercisefoo}{Exercise}

\newtheorem{openfoo}[theoremfoo]{Question}

\newtheorem{nttn}[theoremfoo]{Notation}

\newtheorem{dfntn}[theoremfoo]{Definition}
\newenvironment{definition}{\pagebreak[1]\begin{dfntn}\rm}{\end{dfntn}}

\newenvironment{proof}
    {\pagebreak[1]{\narrower\noindent {\bf Proof:\quad\nopagebreak}}}{\QED}

\newcommand{\floor}[1]{\left\lfloor#1\right\rfloor}
\newcommand{\ceiling}[1]{\left\lceil#1\right\rceil}
\def\nre.{$n$\/-r.e.}

\newtheorem{factfoo}[theoremfoo]{Fact}

\newtheorem{propertyfoo}[theoremfoo]{Property}

\makeatletter

\def\@makechapterhead#1{ \vspace*{50pt} { \parindent 0pt \raggedright
 \ifnum \c@secnumdepth >\m@ne \huge\bf \@chapapp{} \thechapter. \par
 \vskip 20pt \fi \Huge \bf #1\par
 \nobreak \vskip 40pt } }

\def\@sect#1#2#3#4#5#6[#7]#8{\ifnum #2>\c@secnumdepth
     \def\@svsec{}\else
     \refstepcounter{#1}\edef\@svsec{\csname the#1\endcsname.\hskip 1em }\fi
     \@tempskipa #5\relax
      \ifdim \@tempskipa>\z@
        \begingroup #6\relax
          \@hangfrom{\hskip #3\relax\@svsec}{\interlinepenalty \@M #8\par}
        \endgroup
       \csname #1mark\endcsname{#7}\addcontentsline
         {toc}{#1}{\ifnum #2>\c@secnumdepth \else
                      \protect\numberline{\csname the#1\endcsname}\fi
                    #7}\else
        \def\@svsechd{#6\hskip #3\@svsec #8\csname #1mark\endcsname
                      {#7}\addcontentsline
                           {toc}{#1}{\ifnum #2>\c@secnumdepth \else
                             \protect\numberline{\csname the#1\endcsname}\fi
                       #7}}\fi
     \@xsect{#5}}

\def\@begintheorem#1#2{\it \trivlist \item[\hskip \labelsep{\bf #1\ #2.}]}

\def\@opargbegintheorem#1#2#3{\it \trivlist
      \item[\hskip \labelsep{\bf #1\ #2\ (#3).}]}

\makeatother

\newif\ifshortconferences
\shortconferencesfalse
\newif\ifmediumconferences
\mediumconferencesfalse

\def\ending#1{{\count1=#1\relax
\count2=\count1
\divide\count2 by 100
\multiply\count2 by 100
\advance\count1 by -\count2
\ifnum\count1=11
th%
\else \ifnum\count1=12
th%
\else \ifnum\count1=13
th%
\else
\count2=\count1
\divide\count1 by 10
\multiply\count1 by 10
\advance\count2 by -\count1
\ifnum\count2=1
st%
\else \ifnum\count2=2
nd%
\else \ifnum\count2=3
rd%
\else th%
\fi\fi\fi\fi\fi\fi
}}

\def\Proceedingsofthe{\ifshortconferences Proc.\else\ifmediumconferences Proc.\else Proceedings of the\fi\fi}

\newcounter{confnum}

\def\conf#1#2{%
\setcounter{confnum}{#2}%
\addtocounter{confnum}{-\csname #1zero\endcsname}%
\ifnum\value{confnum}=1%
\expandafter\ifx\csname #1One\endcsname\relax%
\Proceedingsofthe\ \arabic{confnum}\ending{\value{confnum}}\ \csname #1name\endcsname%
\else \csname #1One\endcsname\fi%
\else%
\Proceedingsofthe\
\arabic{confnum}\ending{\value{confnum}}\ \csname #1name\endcsname\fi}

\def\qsym{\vrule width0.7ex height0.9em depth0ex}

\newif\ifqed\qedtrue

\def\noqed{\global\qedfalse}

\def\qed{\ifqed{\penalty1000\unskip\nobreak\hfil\penalty50
\hskip2em\hbox{}\nobreak\hfil\qsym
\parfillskip=0pt \finalhyphendemerits=0\par\medskip}\fi\global\qedtrue}

\makeatletter
\def\eqnqed{\noqed
    \def\@tempa{equation}
    \ifx\@tempa\@currenvir\def\@eqnnum{\qsym}%
    \addtocounter{equation}{-1}\else%
    \def\@@eqncr{\let\@tempa\relax
    \ifcase\@eqcnt \def\@tempa{& & &}\or \def\@tempa{& &}%
      \else \def\@tempa{&}\fi
     \@tempa {\def\@eqnnum{{\qsym}}\@eqnnum}
     \global\@eqnswtrue\global\@eqcnt\z@\cr}\fi}

\def\eqnlabel#1#2{\if@filesw {\let\thepage\relax%
   \def\protect{\noexpand\noexpand\noexpand}%
   \edef\@tempa{\write\@auxout{\string
      \newlabel{#2}{{{#1}}{\thepage}}}}%
   \expandafter}\@tempa%
   \if@nobreak \ifvmode\nobreak\fi\fi\fi%
    \def\@tempa{equation}
    \ifx\@tempa\@currenvir\def\theequation{{#1}}%
    \addtocounter{equation}{-1}\else%
    \def\@@eqncr{\let\@tempa\relax
    \ifcase\@eqcnt \def\@tempa{& & &}\or \def\@tempa{& &}%
      \else \def\@tempa{&}\fi
     \@tempa {\def\@eqnnum{{#1}}\@eqnnum}
     \global\@eqnswtrue\global\@eqcnt\z@\cr}\fi}

\makeatother

\def\QED{\qed}

\makeatother

\usepackage{subcaption}

\newcommand{\bigO}{{\rm O}}

\newcounter{example}[section]

\begin{document}

\date{}

\title{Polyhedra Circuits and Their Applications
\thanks{This research is supported in part by National Science
Foundation Early Career Award 0845376 and Bensten Fellowship of the
University of Texas - Rio Grande Valley.} }

	\author{
	Bin Fu$^1$, Pengfei Gu$^1$, and Yuming Zhao$^2$
	\\\\
	$^1$Department of Computer Science\\
	University of Texas - Rio Grande Valley, Edinburg, TX 78539, USA\\
	\\ 
	$^2$School of Computer Science\\
	Zhaoqing University, Zhaoqing, Guangdong 526061, P.R. China\\
} \maketitle

\centerline {Preliminary Version}

\begin{abstract}
     We introduce polyhedra circuits. Each polyhedra circuit characterizes a geometric region in $\mathbb{R}^d$. They can be applied to represent a rich class of geometric objects, which include all polyhedra and the union of a finite number of polyhedra.  They can be used to approximate a large class of $d$-dimensional manifolds in $\mathbb{R}^d$. Barvinok~\cite{Barvinok93} developed polynomial time algorithms to compute the volume of a rational polyhedra, and to count the number of lattice points in a rational polyhedra in a fixed dimensional space $\mathbb{R}^d$ with a fix $d$. Define $T_V(d,\, n)$ be the polynomial time in $n$ to compute the volume of one rational polyhedra, $T_L(d,\, n)$ be the polynomial time in $n$ to count the number of lattice points in one rational polyhedra with $d$ be a fixed dimensional number, $T_I(d,\, n)$ be the polynomial time in $n$ to solve integer linear programming time with $d$ be the fixed dimensional number, where $n$ is the total number of linear inequalities from input polyhedra. We develop algorithms to count the number of lattice points in the geometric region determined by a polyhedra circuit in $\bigO\left(nd\cdot r_d(n)\cdot T_V(d,\, n)\right)$ time and to compute the volume of the geometric region determined by a polyhedra circuit in $\bigO\left(n\cdot r_d(n)\cdot T_I(d,\, n)+r_d(n)T_L(d,\, n)\right)$ time, where $n$ is the number of input linear inequalities, $d$ is number of variables and $r_d(n)$ be the maximal number of regions that $n$ linear inequalities with $d$ variables partition $\mathbb{R}^d$.
     The applications to continuous polyhedra maximum coverage problem, polyhedra maximum lattice coverage problem, polyhedra $\left(1-\beta\right)$-lattice set cover problem, and $\left(1-\beta\right)$-continuous polyhedra set cover problem are discussed. We also show the NP-hardness of the continous maximum coverage problem and set cover problem when each set is represented as union of polyhedra.
     
	 $\textbf{Keywords:}$ Lattice points, Volume, Polyhedra circuits, Union
\end{abstract}

\section{Introduction}
Polyhedra is an important topic in mathematics, and has close
connection to theoretical computer science. There are two natural topics in polyhedra, computing the volume and counting the number of lattice points.  

The problem of counting the number of lattice points in a polyhedra has a wide variety of applicaitons in many areas, for example, number theory,
combinatorics, representation theory, discrete optimization, and
cryptography. It is related to the following problem: given a
polyhedra $P$ which is given by a list of its vertices or by a list
of linear inequalities, our goal is to compute the number $|P\cap
\mathbb{Z}^d|$ of lattice points in $P.$ Researchers have paid much
attention to this problem. Ehrhart~\cite{Ehrhart67} introduced
Ehrhart polynomials that were a higher-dimensional generalization of
Pick's theorem in the Euclidean plane. Dyer~\cite{Dyer91} found
polynomial time algorithms to count the number of lattice points in
polyhedra when dimensional number $d=3,\,4.$
Barvinok~\cite{Barvinok94, Barvinok99} designed a polynomial time
algorithm for counting the number of lattice points in polyhedra when the
dimension is fixed. The main ideas of the algorithm were using
exponential sums~\cite{Brion88, Brion92} and decomposition of
rational cones to primitive cones~\cite{Stanley97} of the polyhedra.
Dyer and Kannan~\cite{Dyer97} simplied Barvinok's polynomial time
algorithm and showed that only very elementary properties of
exponential sums were needed to develop a polynomial time algorithm.
De Loera et al.~\cite{Loera04} described the first implementation of
Barvinok's algorithm called LattE. Some other algebraic-analytic
algorithms have been proposed by many authors (for example,
see~\cite{Baldoni-Silva01, Beck00, Beck03, Lasserre02, MacMahon60,
Pemantle01}.)

Computing exactly the volume of a polyhedra is a basic problem that
have drawn lots of researchers' attentions~\cite{Lasserre83,
Cohen79, Hohenbalken81, Allgower86,Lawrence91,Barvinok93}. It is known that this problem is $\#P$-complete if the polyhedra is given by its vertices or by its facets~\cite{Dyer88, Khachiyan89}. Cohen and
Hickey~\cite{Cohen79} and Von Hohenbalken~\cite{Hohenbalken81}
proposed to compute the volume of a polyhedra by trianguating and
summing the polyhedra. Allgower and Schmidt~\cite{Allgower86}
trianguated the boundary of the polyhedra to compute the volume of
polyhedra. Lasserre~\cite{Lasserre83} developed a recursive method
to compute the volume of polyhedra. Lawrence~\cite{Lawrence91}
computed the volume of polyhedra based on Gram's relation for
polyhedra. Barvinok~\cite{Barvinok93} developed a polynomial-time
algorithm to compute the volume of polyhedra by using the
exponentional integral.

{\bf Motivation:} The existing algorithms related to count the number of lattice points and to compute the volume of the polyhedra concerned with only one polyhedra.  In order to have broader applications, it is essential to develop algorithms to deal with geometric objects that can be generated by a list of polyhedra via unions, intersections, and complementations. In this paper, we propose polyhedra circuits. Polyhedra circuits can be used to represent a large class of geometric objects.  We propose algorithms to compute the volume of the geometric regions by polyhedra circuits, and to count the number of lattice points in the geometric regions by polyhedra circuits. 

{\bf Contributions:}  We have the following contributions to polyhedra. $1.$ We introduce polyhedra circuits. Each polyhedra circuit characterizes a geometric region in $\mathbb{R}^d$. They can be applied to represent a rich class of geometric objects, which include all polyhedra and the union of a finite number of polyhedra.  They can be used to approximate a large class of $d$-dimensional manifolds in $\mathbb{R}^d$. $2.$ We develop an algorithm to compute the volume of the geometric region determined by a polyhedra circuit. We also develop an algorithm to count the number of lattice points of the geometric region determined by a polyhedra circuit.   Our method is based on Barvinok's algorithm, which is only suitable for polyhedra. $3.$ We apply the methods for polyhedra circuits to support new greedy algorithms for the maximum coverage problem and set cover problem, which involve more complex geometric objects. The existing research results about the geometric maximum coverage problem and set cover problem only handle simple objects such as balls and rectangular shapes. All of our algorithms run in polynomial time in $\mathbb{R}^d$ with a fixed $d.$


{\bf Organization:} The rest of paper is organized as follows. In
Section~\ref{pre-sec}, we introduce some important theorems and definitions. In Section~\ref{basic-sec}, we introduce polyhedra circuits and develop polynomial time algorithms to compute the volume of a given list of linear inequalities and to count the number of lattice points in a given list of linear inequalities. In Section~\ref{maxi-cover-app}, we apply the algorithm to polyhedra maximum coverage probelm. In Section~\ref{maxi--lattice-cover-app}, we apply the aglorithm to polyhedra maximum lattice coverage problem. In Section~\ref{set_cover_sec}, we apply the algorithm to polyhedra $\left(1-\beta\right)$-lattice set cover problem. In Section~\ref{conti-set-cover-app}, we apply the algorithm to $\left(1-\beta\right)$-continuous polyhedra set cover problem. Section~\ref{NP-hardness} shows the NP-hardness of the maximum coverage problem and set cover problem when each set is represented as union of polyhedrons. Section~\ref{concl} summarizes the conclusions.

\section{Preliminaries}\label{pre-sec}
In this section, we introduce some definitions and some important theorems that play a basic role in our algorithms.

\begin{definition}[\cite{Barvinok94}]\label{polyhedra-def}
	A {\it rational polyhedra } $P\in \mathbb{R}^d$ is a set defined by finitely many linear inequalities:
	$$P=\bigg\{x: \quad \sum\limits_{j=1}^{d}a_{ij}\xi_j\leq b_i\quad\text{for}\quad x=\left(\xi_1,\, \cdots,\,\xi_d\right)\quad\text{and}\quad i\in I\bigg\}$$
    for some finite (possibly empty) set $I,$ where $a_{ij}$ and $b_i$ are integer numebrs.
\end{definition}
Note that we always assume that the polyhedra that we deal with are rational polyhedra throughout the paper.
\begin{theorem}[\cite{Barvinok93}]\label{Barvinok-volume-theo}
	There exists a polynomial time algorithm for computing the volume of one rational polyhedra. 
\end{theorem}

\begin{theorem}[\cite{Barvinok94}]\label{Barvinok-lattice-point-theo}
	Let us fix $d\in \mathbb{N}.$ Then there exists a polynomial time algorithm for counting lattice points in one rational polyhedra when the dimension $d$ is fixed. 
\end{theorem}


\begin{definition}[\cite{Schrijver86}]\label{interger-linear-programming-def}
	{\it Integer Linear Programming} problem is a constrained optimization problem of the form:
	$$\max \big\{cx\ |\ Ax\leq b\,;\ x\ \text{integral}\big\},$$
	where $A$ is a given $n\times d$ matrix, $b$ an $n$-vector and $c$ a $d$-vector, and the entries of $A,\, b$ and $c$ are rational numbers.
\end{definition}


\begin{theorem}[\cite{Clarkson95}]\label{integer-linear-programming}
	Let us fix $d\in \mathbb{N}.$ Then there exists a las vegas algorithm to solve the integer linear programming problem when the dimension $d$ is small. 
\end{theorem}


\begin{definition}
	Let us fix $d\in \mathbb{N}.$ Define $T_V(d,\, n)$ be the polynomial time in $n$ to compute the volume of one rational polyhedra by Theorem~\ref{Barvinok-volume-theo}, $T_L(d,\, n)$ be the polynomial time in $n$ to count the number of lattice points in one rational polyhedra with $d$ be a fixed dimensional number by Theorem~\ref{Barvinok-lattice-point-theo}, $T_I(d,\, n)$ be the polynomial time in $n$ to solve integer linear programming time with the fixed dimensional number $d$ by Theorem~\ref{integer-linear-programming}, where $n$ is the total number of linear inequalities from input polyhedra.
\end{definition}

\section{Algorithms about Polyhedra Circuits}\label{basic-sec}
In this section, we give the definition of polyhedra circuits and develop algorithms to compute the volume of a region given by a polyhedra circuit and to count the number of lattice points inside a region given by a polyhedra circuit. We always assume that the linear inequalities that we deal with are linear inequalities in $\mathbb{R}^d$ and the coefficient of each linear inequalities are integers throughout the paper, which means the polyhedra that we deal with are rational polyhedra throughout the paper.

\begin{definition}
	A {\it hyper plane} in $\mathbb{R}^d$ can be defined by a linear equation of the form 
	$$a_1x_1+a_2x_2+\cdots+a_dx_d=b,$$
	where $a_1,\, a_2,\,...,\, a_d$ and $b$ are constants.
\end{definition}

\begin{definition}
Define $r_d(n)$ to be the maximal number of regions that $n$ hyper planes partition $\mathbb{R}^d$.
\end{definition}

Lemma~\ref{num-of-regions} gives the lower bound and upper bound of the maximal number $r_d(n)$ of regions that $n$ hyper planes partition $\mathbb{R}^d$.
\begin{lemma}\label{num-of-regions}
 $\left(\floor{n\over d}\right)^d\le r_d(n)\le {n^d\over d!}+(n+1)^{d-1}$ regions.
\end{lemma}

\begin{proof}
We have the recursion $r_{d+1}(n+1)=r_{d+1}(n)+r_d(n)$. This is
because the $n+1$-th $d$-dimensional hyper plane can have at most
$r_d(n)$ regions in it, and increases the number of regions by at
most $r_d(n)$. It is trivial $r_1(n)\le n+1$, and also trivial for
$n=0$ or $1$ for all $d\ge 1$. It follows a simple induction. Assume
that $r_d(n)\le {n^d\over d!}+(n+1)^{d-1}$ and $r_{d+1}(n)\le
{n^{d+1}\over (d+1)!}+(n+1)^d$.

\begin{eqnarray*}
r_{d+1}(n+1)&=& r_{d+1}(n)+r_d(n)\\
&\le& ({n^{d+1}\over
(d+1)!}+(n+1)^d)+({n^d\over d!}+(n+1)^{d-1})\\
&=& {n^{d+1}+(d+1)n^d\over (d+1)!}+((n+1)^d+(n+1)^{d-1})\\
&<& {(n+1)^{d+1}\over (d+1)!}+(n+2)^d.
\end{eqnarray*}

For the left side of the inequalities, we just let there are
${n\over d}$ hyper planes to be parallel to each of $d$ axis in
$\mathbb{R}^d$. Therefore, we have the lower bound $\left(\floor{n\over d}\right)^d$.

\end{proof}

\begin{definition}
For a set of linear inequalities in $\mathbb{R}^d$, a region is {\it atomic}
if it is formed by a subset of linear inequalities, and does not
contain any proper subregion that can be formed by another subset of
linear inequalities.
\end{definition}
\begin{figure}[H]
	\centering
	\includegraphics[width=0.5\textwidth]{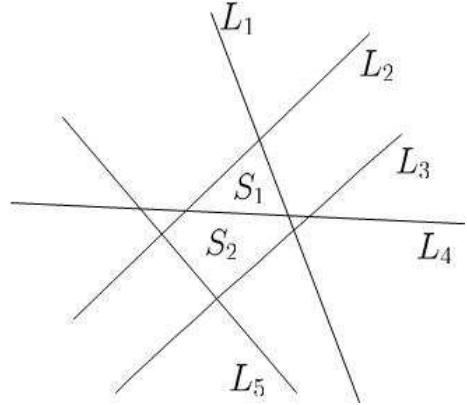} 
	\caption{Example of atomic regions with $S_1$ and $S_2$ are atomic regions, but $S_1\cup S_2$ is not atomic region, although it is formed by a subset linear inequalities.}
\end{figure}

\begin{definition}\label{circuit-def}
	{\it A polyhedra circuit} consists is a circuit that consists of multiple  layers of gates:
	\begin{itemize}
		\item Each input gate is a linear inequality of format  $\sum\limits_{i=1}^d a_ix_u\le b$ or $\sum\limits_{i=1}^d a_ix_u< b$ (represent a half space in $\mathbb{R}^d$), where $b, a_1,\cdots,a_d$ are integers, 
		
		\item Each internal gate is either union or intersection operation, and
		
		\item The only output gate is on the top of the circuit.
	\end{itemize}
	
\end{definition}
\begin{figure}[H]
	\centering
	\includegraphics[width=0.8\textwidth]{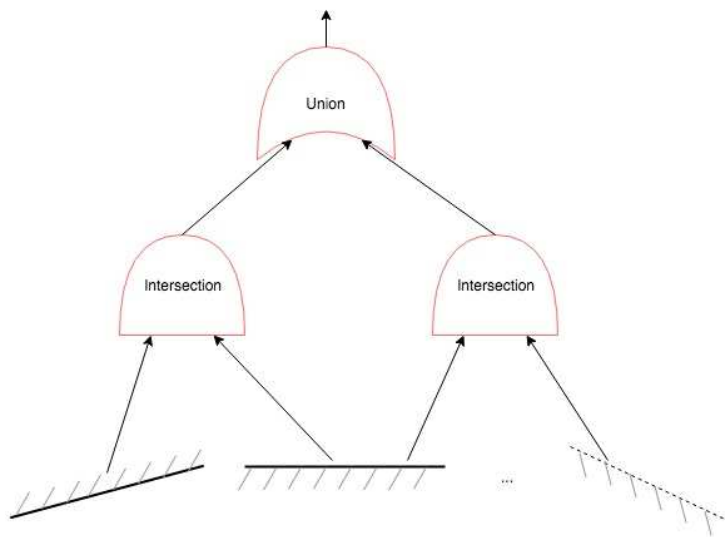} 
	\caption{Example of a Polyhedra Circuit}
\end{figure}

For example, the expression $\left(\left(-x<-1\right)\wedge\left(x\leq 3\right)\right)\vee\left(\left(-x\leq -2\right)\wedge\left(x\leq 5\right)\right)$ is $\mathbb{R}^1$ polyhedra circuit to express to the union of $(1,\, 3]$ and $[2,\, 5].$ Its output is the region $(1,\, 5]$ in $\mathbb{R}^1.$

Note that we do not consider the negation operation in our circuit gates, since a negation operation in a gate can be handled by taking the negation of the input linear inequalities via De Morgan's laws. We note also that a union of multiple polyhedra can produce very complicate shape in a $d$-dimensional space for $d\geq 2.$

Lemma~\ref{circuit-union-lemma} shows that the output of a polyhedra circuit is a disjoint union of atomic regions.
\begin{lemma}\label{circuit-union-lemma}
	The region generated by a polyhedra circuit is a disjoint union of
	atomic regions via the hyper planes of the input polyhedra.
\end{lemma}

\begin{proof}
	It can be proven via an induction on the depth of the circuits. It is
	trivial when the circuits has depth zero. Assume that there are $G_1,\,\cdots,\, G_k$ gates in depth $k$, and each $G_i$
	generates a region to be the union of atomic regions for $i=1,\,\cdots,\, k$. We need to show that each gate $G_j$ in depth $k+1$ generates a region to be the union of atomic regions.
	
	The union or intersection of the union of atomic regions also be the union of atomic regions. Every gate is either union or intersection via Definition~\ref{circuit-def}, therefore, every gate $G_j$ in depth $k+1$ generates a region to be the union of atomic regions.

	Assume that $G$ is the
	output gate. 
	By the induction, we can say that region from $G$ is also the union of atomic regions.
\end{proof}

Lemma~\ref{point-atomic_thm} shows that the volume of rational polyhedra formed by $Ax< b$ equals the volume of rational polyhedra formed by $Ax\leq b.$
\begin{lemma}\label{point-atomic_thm}
	The volume of atomic region formed by $Ax< b$ equals the volume of atomic region formed by $Ax\leq b,$ where $A$ is a given $n\times d$ matrix, $b$ is an $n$-dimensional vector
\end{lemma}
\begin{proof}
	Since the boundaries of the atomic region formed by $Ax\leq b$ have no volume, then the volume of atomic region that is formed by $Ax< b$ equals the volume of atomic region that is formed by $Ax\leq b,$ where $A$ is a given $n\times d$ matrix and $b$ is an $n$-dimensional vector.
\end{proof}

Lemma~\ref{integer-point-atomic_thm} shows that the number of lattice points in rational polyhedra formed by $Ax< b$ equals the the number of lattice points in rational polyhedra formed by $Ax\leq b-I,$ where all the elements of $n\times d$ matrix $A$ and $n$-dimensional vector $b$ are inegers, and $I$ is $n$-dimensional vector with all $1s.$

\begin{lemma}\label{integer-point-atomic_thm}
	The number of lattice points in atomic region formed by $Ax< b$ equals the number of lattice points in atomic region formed by $Ax\leq c,$ where $A$ is a given $n\times d$ matrix with integer elements, $b$ is an $n$-dimensional vector with integer elements and $c$ is an $n$-dimensional vector with $c=b-I$ for $n$-dimensional vector $I$ whose elements are all $1s$.
\end{lemma}
\begin{proof}
	The number of lattice points in linear inequality $a_1x_1+\cdots a_dx_d<b_1$ is the same as the number of lattice points in linear inequality $a_1x_1+\cdots a_dx_d<b_1-1,$ becasue $a_1,\,...,\,a_d$ and $b_1$ are integers and each $x_i$ will take an integer value. Therefore, the number of lattice points in atomic region formed by $Ax< b$ equals the number of lattice points in atomic region formed by $Ax\leq c,$ where $A$ is a given $n\times d$ matrix with integer elements, $b$ is an $n$-dimensional vector with integer elements and $c$ is an $n$-dimensional vector with $c=b-I$ for $n$-dimensional vector $I$ whose elements are all $1s$.
\end{proof}

Lemma~\ref{atomic_decomposition_thm} shows that a list of
linear inequalities partition space $\mathbb{R}^d$ into all of the atomic regions with an interior point in each atomic region.
\begin{lemma}\label{atomic_decomposition_thm}
There is a $\bigO\left(nd\cdot r_d(n)\cdot T_V(d,\, n)\right)$ time algorithm such that given a list of
$n$ linear inequalities in $\mathbb{R}^d$, it produces all of the atomic regions with one interior point
from each of them, where $T_V(.)$ is the time to compute the volume of one
rational polyhedra and $r_d(n)$ be the maximal number of regions that $n$ hyper planes partition $\mathbb{R}^d$.
\end{lemma}

\begin{proof}
It can be proven via an induction method. Let $L_1,\,
L_2,\,\cdots,\, L_n$ be the linear inequalities. If there is only one linear inequality, say $L_1,$ $L_1$ partitions $\mathbb{R}^d$ into two atomic regions $S_1$ and $S_2,$ where $S_1$ is formed by $L_1$ and $S_2$ is formed by $\neg L_1,$ it is easy to generate two interior points $p_1\in S_1$ and $p_2\in S_2.$

Assume that we have obtained atomic regions $S_1,\, S_2,\,\cdots,\, S_t$
that are formed by $L_1,\, \cdots,\, L_k$, and each $S_i$ has an interior
point $p_i\in S_i$ for $i=1,\,\cdots, t$. We need to show that we can obtian all of the
atomic regions that are formed by $L_1,\,\cdots,\, L_k,\, L_{k+1},$ and one interior point from each of them.

For each $S_i$, assume that $S_i$ is formed by hyper planes $L'_{1},\,\cdots,\, L'_{s}$ with $s\leq k.$ Consider $L_{k+1},$ and the volume of $S'_{i,\,1}$ formed by adding $L_{k+1}$ to hyper planes $L'_{1},\,\cdots,\, L'_{s}$ and $S'_{i,\,2}$ formed by adding $\neg L_{k+1}$ to hyper planes $L'_{1},\,\cdots,\, L'_{s}.$ Note that we compute the volume of $S'_{i,\,1}$ and $S'_{i,\,2}$ by using Theorem~\ref{Barvinok-volume-theo}. Note also that we can compute the volume of $S'_{i,\,1}$ and $S'_{i,\,2}$ correctly becasue of Lemma~\ref{point-atomic_thm}. 

We discuss two cases:

Case 1: If the volume of $S'_{i,\,1}$ or $S'_{i,\,2}$ equals to $0,$ then $S_i$ is not seperated by $L_{k+1}.$ In this case, the interior point is $p_i\in S_i.$

Case 2: Otherwise, $L_{k+1}$ partitions $S_i$ into two regions $S_{i,\,1}$ and $S_{i,\,2},$ and we have to find the interior point for each regions. In this case, $L_{k+1}$ intersection $S_i$ with a $(d-1)$-dimensional atomic region, let us denote as $B_{d-1}$. Similary, this $(d-1)$-dimensional atomic region $B_{d-1}$ can be partitioned into $(d-2)$-dimensional atomic regions. By using this way, we can easily find a interior point in $0$-dimensional atomic region, which is an interior point $p_j$ of $B_{d-1}.$   

Construct line $l=p_j+tv$ perpendicular to $L_{k+1}$ where $v$ is the vector that is perpendicular to $L_{k+1}.$ Find the points $p_h,\, p_z$ that $l$ intersection of the first hyper plane of $S_{i, 1}$ and the first hyper plane of $S_{i, 2},$ respectively (See Fig.~\ref{fig3}), then the interior point of $S_{i, 1}$ is $\frac{p_j+p_h}{2},$ and the interior point of $S_{i, 2}$ is $\frac{p_j+p_z}{2}.$
\begin{figure}[H]
	\centering
	\includegraphics[width=0.5\textwidth]{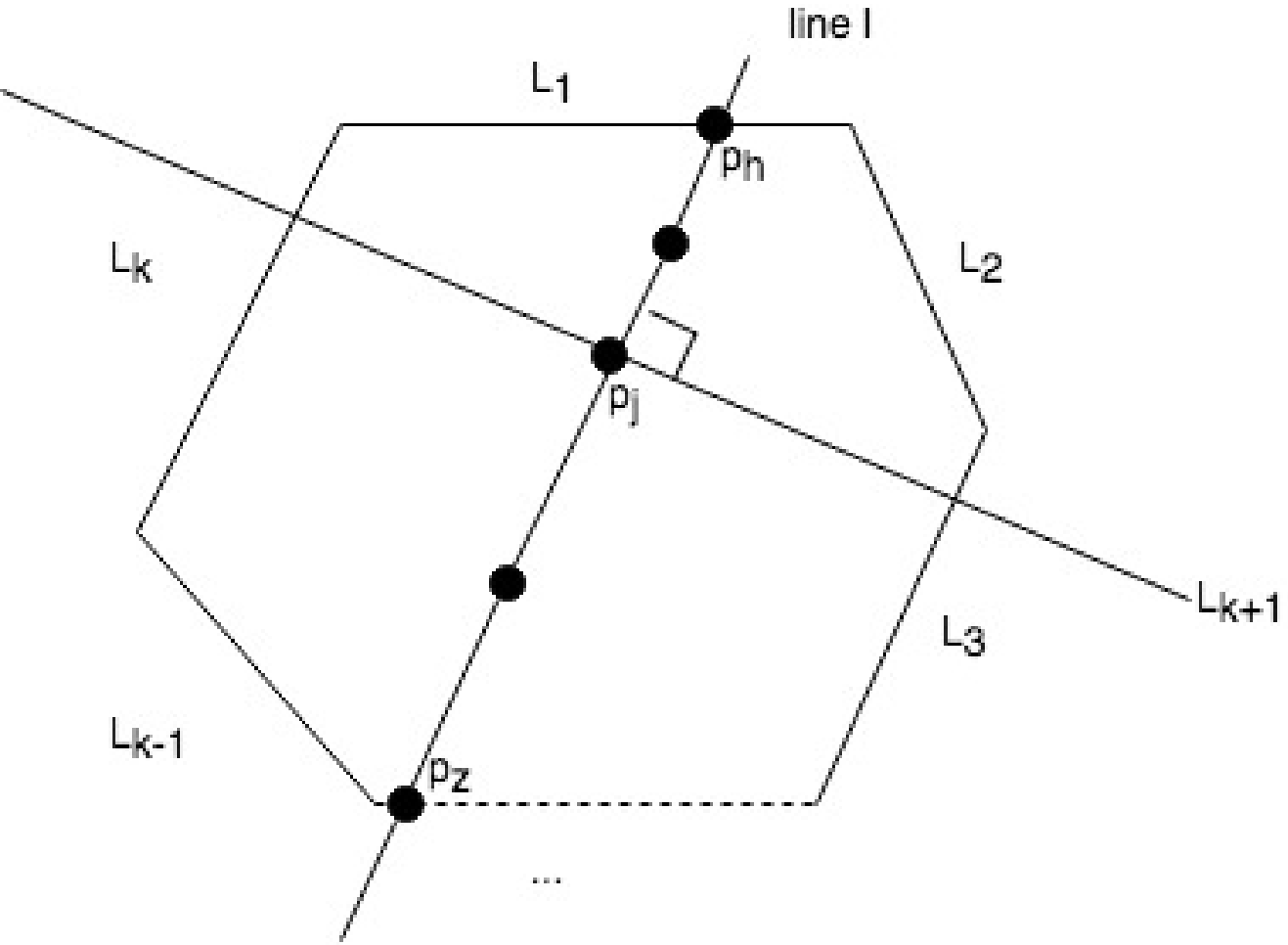} 
	\caption{}\label{fig3}
\end{figure}
Therefore, we have obtained atomic regions $S_1,\, S_2,\,\cdots,\, S_t,\, \cdots,\, S_h$ that are formed by $L_1,\,\cdots,\, L_k,\, L_{k+1}$ for some interger $h>t,$ and one interior point form each of them. 

Define $T_{R}\left(d,\,n\right)$ be the running time to generate an interior point in $d$-dimensional atomic region $S_i$ with $n$ linear inequalities. Then, we have 
\begin{eqnarray*}
	T_{R}\left(d,\,k+1\right)&=&2r_d(k)\cdot T_V\left(d,\, k\right)+T_{R}\left(d-1,\,k\right)\\
	&=& 2r_d(k)\cdot T_V\left(d,\, k\right)+...+2r_1(k)\cdot T_V\left(1,\, k\right)\\
	&\leq&2d\cdot r_d(k)\cdot T_V\left(d,\, k\right),
\end{eqnarray*} 
and
$$\sum\limits_{k=1}^{n-1}T_{R}\left(d,\,k+1\right)\leq 2(n-1)d\cdot r_d(k)\cdot T_V\left(d,\, k\right).$$

Therefore, the running of the algorithm is $\bigO\left(nd\cdot r_d(n)\cdot T_V(d,\, n)\right).$


\end{proof}

Lemma~\ref{atomic_decomposition_integer-thm} shows that a list of
linear inequalities partition space $\mathbb{R}^d$ into all of the atomic regions with one integer point from each of them if exists.
\begin{lemma}\label{atomic_decomposition_integer-thm}
	There is a $\bigO\left(n\cdot r_d(n)\cdot T_I(d,\, n)\right)$ time algorithm such that given a list of
	$n$ linear inequalities in $\mathbb{R}^d$, it produces all of the atomic regions, and one integer point from each of them if exists, where $T_I(.)$ is the time to solve integer linear programming problem and $r_d(n)$ be the maximal number of regions that $n$
	hyper planes partition $\mathbb{R}^d$.
\end{lemma}

\begin{proof}
	It can be proven via an induction method. Let $L_1,\,
	L_2,\,\cdots,\, L_n$ be the linear inequalities. If there is only one linear inequality, say $L_1,$ $L_1$ partitions $\mathbb{R}^d$ into two atomic regions $S_1$ and $S_2,$ where $S_1$ is formed by $L_1$ and $S_2$ is formed by $\neg L_1,$ it is easy to generate two integer points $p_1\in S_1$ and $p_2\in S_2.$
	
	Assume that we have obtained atomic regions $S_1,\, S_2,\,\cdots,\, S_t$
	that are formed by $L_1,\, \cdots,\, L_k$, and each $S_i$ has an integer
	point $p_i\in S_i$ for $i=1,\,\cdots, t$. We need to show that we can obtian all of the
	atomic regions that are formed by $L_1,\,\cdots,\, L_k,\, L_{k+1}$ and one integer point from each of them if exists.
	
	For each $S_i$, assume that $S_i$ is formed by linear inequalities $L'_{1},\,\cdots,\, L'_{s}$ with $s\leq k.$ 
	
	Consider $S_{i,\,1}$ formed by adding $L_{k+1}$ to linear inequalities $L'_{1},\,\cdots,\, L'_{s}$ and $S_{i,\,2}$ formed by adding $\neg L_{k+1}$ to linear inequalities $L'_{1},\,\cdots,\, L'_{s}.$ We use Theorem~\ref{integer-linear-programming} to find if there is an integer point in $S_{i,\,1}$ and $S_{i,\,2}.$ If it exists, we keep the corresponding atomic region.

	Therefore, we have obtained atomic regions $S_1,\, S_2,\,\cdots,\, S_t,\, \cdots,\, S_h$ that are formed by $L_1,\,\cdots,\, L_k,\, L_{k+1}$ for some interger $h>t,$ and one integer point form each of them. 
	
	Define $T_{R}\left(d,\,n\right)$ be the running time to generate an integer point in $d$-dimensional atomic region $S_i$ with $n$ linear inequalities. Then, we have 
	\begin{eqnarray*}
		T_{R}\left(d,\,k+1\right)&=&2r_d(k)\cdot T_I\left(d,\, k\right)
	\end{eqnarray*} 
	and
	$$\sum\limits_{k=1}^{n-1}T_{R}\left(d,\,k+1\right)= 2(n-1)\cdot r_d(k)\cdot T_V\left(d,\, k\right).$$
	
	Therefore, the running of the algorithm is $\bigO\left(n\cdot r_d(n)\cdot T_I(d,\, n)\right).$
	
	
\end{proof}

Theorem~\ref{volume-basic-thm} shows that the algorithm can compute the volume of the region by the polyhedra circuits.
\begin{theorem}\label{volume-basic-thm}
There is a $\bigO\left(nd\cdot r_d(n)\cdot T_V(d,\, n)\right)$ time algorithm such that  given a polyhedra circuit in $\mathbb{R}^d$,
the algorithm computes the volume of the region determined by the input  polyhedra circuit, where $T_V(.)$ is the time to compute the volume of one rational
polyhedra, and $n$ is the number of linear inequalities in the input layer of the polyhedra circuit.
\end{theorem}

\begin{proof}
	By Lemma~\ref{circuit-union-lemma} the region generated by a polyhedra circuits is a disjoint union of atomic regions, by Lemma~\ref{atomic_decomposition_thm} there is an interior point $p_i$ for each atomic region. For each interior point $p_i,$ we check if it is contained by polyhedra circuits, if it doese, then add the volume of the atomic regions which can be computed by Theorem~\ref{Barvinok-volume-theo}.
	
	And the running time of the algorithm is
	$$\bigO\left(nd\cdot r_d(n)\cdot T_V(d,\, n)\right)+r_d(n)T_V(d,\, n)$$ via Theorem~\ref{Barvinok-volume-theo}.
	Therefore the running time $\bigO\left(nd\cdot r_d(n)\cdot T_V(d,\, n)\right).$

\end{proof}

Theorem~\ref{lattice-basic-thm} shows that the algorithm can count the number of lattice points in the region by the polyhedra circuits.
\begin{theorem}\label{lattice-basic-thm}
 There is a $\bigO\left(n\cdot r_d(n)\cdot T_I(d,\, n)+r_d(n)T_L(d,\, n)\right)$ time algorithm such that given a polyhedra circuit in $\mathbb{R}^d$,
the algorithm counts the number of lattice points in the region determined by the input polyhedra circuit, where
$T_L(.)$ is the time to compute the number of lattice points of one rational
polyhedra and $T_I(.)$ is the time to solve integer linear programming problem and and $n$ is the number of linear inequalities in the input layer of the polyhedra circuit.
\end{theorem}

\begin{proof}
	By Lemma~\ref{circuit-union-lemma} the region generated by a polyhedra circuits is a disjoint union of atomic regions, and by Lemma~\ref{atomic_decomposition_integer-thm} there is an integer point $p_i$ for each atomic region. For each integer point $p_i,$ we check if it is contained by polyhedra circuits, if it doese, then add the number of lattice points in the atomic region which can be computed Theorem~\ref{Barvinok-lattice-point-theo}.
	
	And the running time of the algorithm is
	$$\bigO\left(n\cdot r_d(n)\cdot T_I(d,\, n)\right)+r_d(n)T_L(d,\, n)$$ via Theorem~\ref{Barvinok-volume-theo}.
	
	Therefore the running time $\bigO\left(n\cdot r_d(n)\cdot T_I(d,\, n)+r_d(n)T_L(d,\, n)\right).$
	
\end{proof}

An immediate application of the algorithm is to compute the volume of polyhedra
union, and count the number of lattice points in the union of
polyhedra.

\begin{corollary}\label{coro-volum}
    There is a $\bigO\left(nd\cdot r_d(n)\cdot T_V(d,\, n)\right)$ time algorithm to compute the volume of the union of polyhedra when give a list of polyhedra, where $T_V(.)$ is the time to compute the volume of one rational
    polyhedra.
\end{corollary}
\begin{proof}
	Let the output gate of the polyhedra circuitsis is union (See Fig.~\ref{fig4}). Then it follows Theorem~\ref{volume-basic-thm}.
	\begin{figure}[H]
		\centering
		\includegraphics[width=0.5\textwidth]{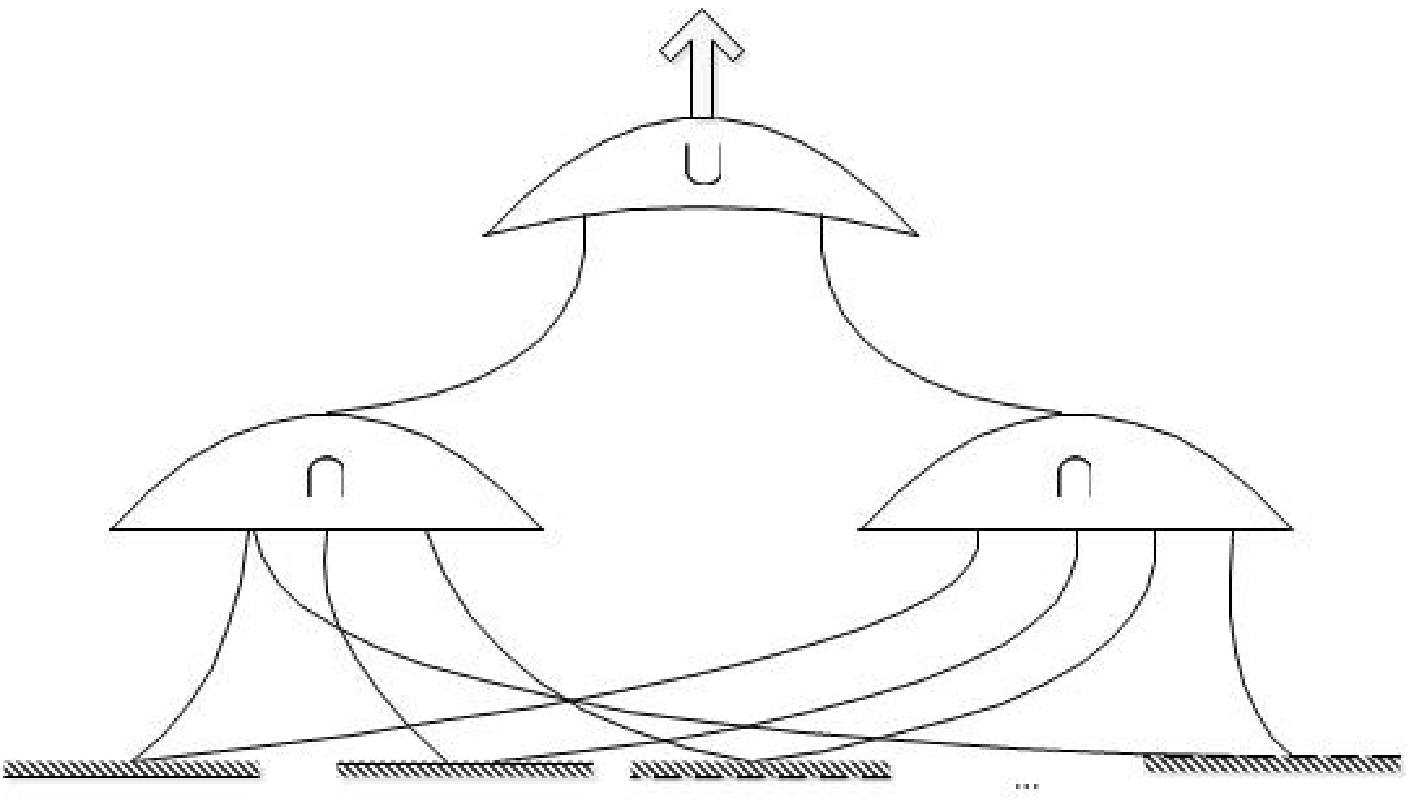} 
		\caption{}\label{fig4}
	\end{figure}
\end{proof}

\begin{corollary}\label{coro-lattice}
    There is a $\bigO\left(n\cdot r_d(n)\cdot T_I(d,\, n)+r_d(n)T_L(d,\, n)\right)$ time algorithm to count the number of lattice points in the union of polyhedra when give a list of polyhedra, where $T_L(.)$ is the time to compute the number of lattice points of one rational
    polyhedra and $T_I(.)$ is the time to find a integer point in one rational polyhedra.
\end{corollary}
\begin{proof}
	Let the output gate of the polyhedra circuitsis is union (See Fig.~\ref{fig5}). Then it follows Theorem~\ref{lattice-basic-thm}.
	\begin{figure}[H]
		\centering
		\includegraphics[width=0.5\textwidth]{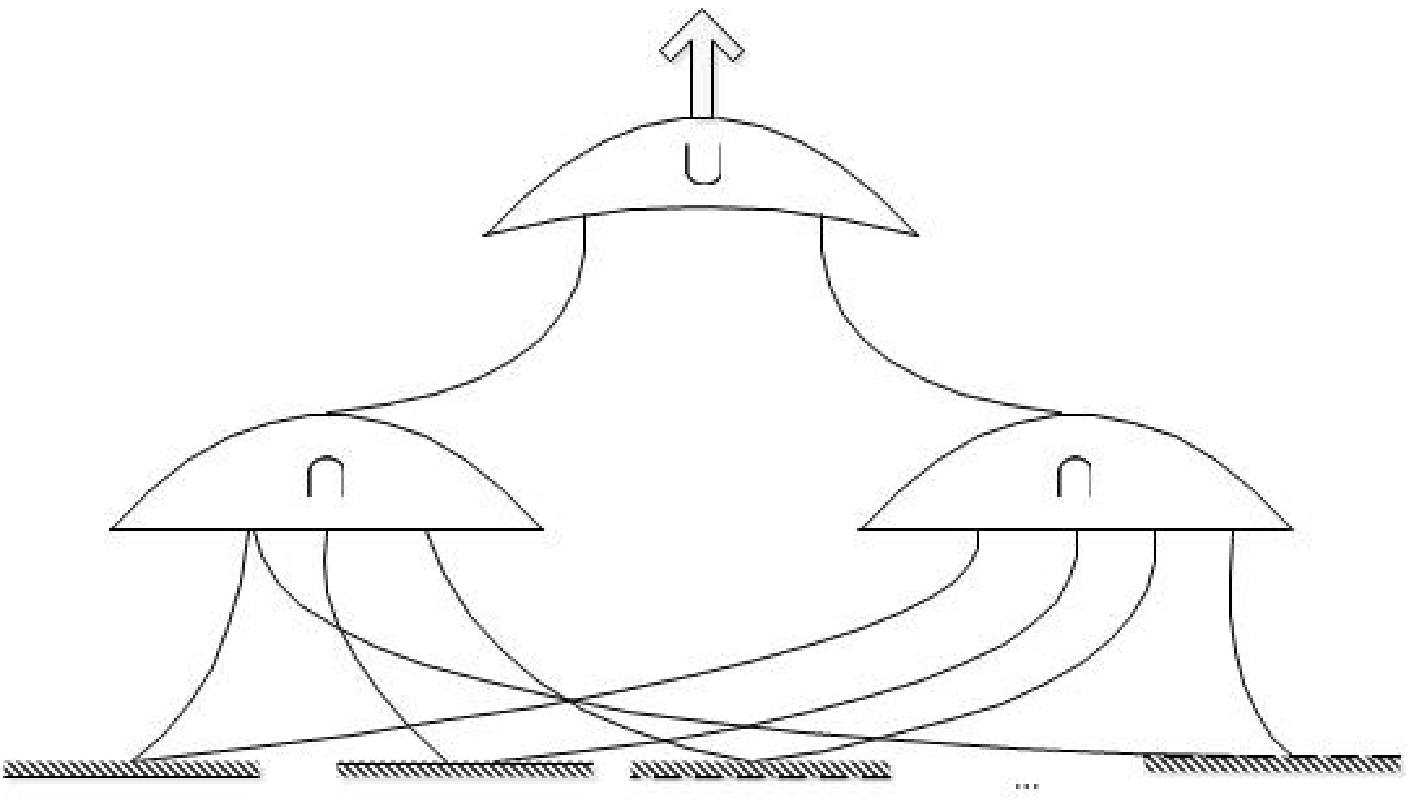}
		\caption{}\label{fig5}
	\end{figure}
\end{proof}

\begin{definition}\label{def20} 
   	Define $T_{U,\,L}(d,\,n)$ be the running time to count the number of lattice points in the union of rational polyhedra by Corollary~\ref{coro-lattice} and $T_{U,\,V}(d,\,n)$ be the running time to compute the volume of the union of rational polyhedra by Corollary~\ref{coro-volum} when given a list of polyhedra where $n$ is the total number of linear inequalities from input polyhedra and $d$ is number of variables of linear inequality.
\end{definition}

\section{Application in Continuous Polyhedra Maximum Coverage Problem}\label{maxi-cover-app}

In this section, we show how to apply the method at
Section~\ref{basic-sec} to the continuous polyhedra maximum coverage problem. Before present the algorithm, we give some definitions about the continuous polyhedra maximum coverage problem.

\begin{definition} {\it Continuous Maximum Coverage Problem:}
Given an integer $k$, a set of regions $S_1,
\,S_2,\,\cdots,\, S_m$, select $k$ regions $S_{i_1},\,\cdots,\, S_{i_k}$ such
that $S_{i_1}\cup S_{i_2}\cup\cdots\cup S_{i_k}$ has the
maximum volume of $S_{1}\cup S_{2}\cup\cdots\cup S_{m}$.
\end{definition}

\begin{definition} {\it Continuous Polyhedra Maximum Coverage Problem:}
	It is a continuous maximum coverage problem when $S_1,\,\cdots,\, S_m$ are all polyhedra.
\end{definition}

\begin{theorem}
There is a $\bigO\left(km\cdot T_{U,\,V}(d,\,n)\right)$ time approximation algorithm for
the continuous polyhedra maximum coverage problem with approximation ratio
$\left(1-{1\over e}\right)$, where $T_{U,\,V}(d,\,n)$ be the running time to compute the volume of the union of polyhedra with $n$ be the total number of linear inequalities from input polyhedra.
\end{theorem}

\begin{proof}
It follows from greedy algorithm and Theorem~\ref{volume-basic-thm}. Let $vol(P)$ denotes the volume of region $P$ and let $V$ denotes $vol\left(S_{1}\cup\cdots \cup S_{m}\right).$ The greedy method works by picking, at each stage, the algorithm select region S that covers the maximum volume that are uncovered.
Assume that $t$ regions have been selected, say, $S_{i_1},\, S_{i_2},\,...,\,S_{i_t},$ the greedy approach needs to select $S_j$ such that $vol\left(S_j- \left(S_{i_1}\cup S_{i_2}\cup...\cup S_{i_t}\right)\right)$ is maximal for the next choice.

Since 
\begin{eqnarray*}
	vol\left(S_j- \left(S_{i_1}\cup S_{i_2}\cup...\cup S_{i_t}\right)\right)&=&vol\left(S_j\cup S_{i_1}\cup S_{i_2}\cup...\cup S_{i_t}\right)\\
	&& -vol\left(S_{i_1}\cup S_{i_2}\cup...\cup S_{i_t}\right),
\end{eqnarray*} 
and we can compute $vol\left(S_j\cup S_{i_1}\cup S_{i_2}\cup...\cup S_{i_t}\right)$ and $vol\left(S_{i_1}\cup S_{i_2}\cup...\cup S_{i_t}\right)$ via Theorem~\ref{volume-basic-thm} and Corollary~\ref{coro-volum}, then at each stage, we just need to traversal $m$ polyhedra to find $S_j.$

Let $V(l)$ be the volume of coverage for $l$ regions, and let $OPT$ be the maximum volume of union of a optimal solution $S_{i_1},\, S_{i_2},\,...,\,S_{i_k}$ for the continuous polyhedra maximum coverage problem. 

First, we show that 
$$V(t)\geq \left(1-\left(1-\frac{1}{k}\right)^t\right)OPT$$
by using induction method.

It is trivial at $t=1,$ since $V(1)\geq \frac{OPT}{k}.$ Assume $V(t)\geq \left(1-\left(1-\frac{1}{k}\right)^t\right)OPT.$ Consider the case $t+1.$
\begin{eqnarray*}
	V(t+1)&=&V(t)+vol\left(S_{i_{t+1}}- \left(S_{i_1}\cup S_{i_2}\cup...\cup S_{i_t}\right)\right)\\
	&\geq& V(t)+\frac{OPT-V(t)}{k}\\
	&=&\left(1-\frac{1}{k}\right)V(t)+\frac{OPT}{k}\\
	&\geq&\left(1-\left(1-\frac{1}{k}\right)^t\right)\left(1-\frac{1}{k}\right)OPT+\frac{OPT}{k}\\
	&=&\left(1-\left(1-\frac{1}{k}\right)^{t+1}\right)OPT.
\end{eqnarray*} 

Therefore,
$$V(t)\geq \left(1-\left(1-\frac{1}{k}\right)^t\right)OPT.$$

Since $\left(1-\frac{1}{k}\right)^k$ is increasing and $\frac{1}{e}=\left(1-\frac{1}{k}\right)^k+\Omega\left(\frac{1}{k}\right),$ therefore,
$$V(k)\geq \left(1-\left(1-\frac{1}{k}\right)^k\right)OPT\geq \left(1-\frac{1}{e}\right)OPT.$$

There are $k$ stages. At each stage, the algorithm needs to check $m$ polyhedra in order to select the next region that covers the maximum volume that are uncovered. By Definition~\ref{def20}, $T_{U,\,V}(d,\,n)$ is the running time to compute the volume of the union of polyhedra, therefore the running time of the algorithm is $\bigO\left(km\cdot T_{U,\,V}(d,\,n)\right).$

\end{proof}

\section{Application in Polyhedra Maximum Lattice Coverage Problem}\label{maxi--lattice-cover-app}

In this section, we show how to apply the method at
Section~\ref{basic-sec} to the polyhedra maximum lattice coverage problem. Before present the algorithm, we give some definitions about polyhedra maximum lattice coverage problem.

\begin{definition} {\it Maximum Lattice Coverage Problem:}
	Given an integer $k$, a set of regions $S_1,\,
	S_2,\,\cdots,\, S_m$, select $k$ regions $S_{i_1},\,\cdots,\, S_{i_k}$ such that $S_{i_1}\cup S_{i_2}\cup\cdots\cup S_{i_k}$ has the
	maximum number of lattice points of $S_{1}\cup S_{2}\cup\cdots\cup S_{m}$.
\end{definition}

\begin{definition} {\it Polyhedra Maximum Lattice Coverage Problem:}
	It is a maximum lattice coverage problem when $S_1,\,\cdots,\, S_m$ are all polyhedra.
\end{definition}

\begin{theorem}
	There is a $\bigO\left(km\cdot T_{U,\,L}(d,\,n)\right)$ time approximation algorithm for
	the polyhedra maximum lattice coverage problem with approximation ratio
	$\left(1-{1\over e}\right)$, where $T_{U,\,L}(d,\,n)$ is the running time to count the number of lattice points in the union of polyhedra with $n$ be the total number of linear inequalities from input polyhedra.
\end{theorem}

\begin{proof}
	It follows from greedy algorithm and Theorem~\ref{lattice-basic-thm}. Let $L(P)$ denotes the volume of region $P$ and let $V$ denotes $L\left(S_{1}\cup\cdots \cup S_{m}\right).$ The greedy method works by picking, at each stage, the algorithm select region S that covers the maximum volume that are uncovered.
	Assume that $t$ regions have been selected, say, $S_{i_1},\, S_{i_2},\,...,\,S_{i_t},$ the greedy approach needs to select $S_j$ such that $L\left(S_j- \left(S_{i_1}\cup S_{i_2}\cup...\cup S_{i_t}\right)\right)$ is maximal for the next choice.
	
	Since 
	\begin{eqnarray*}
		L\left(S_j- \left(S_{i_1}\cup S_{i_2}\cup...\cup S_{i_t}\right)\right)&=&L\left(S_j\cup S_{i_1}\cup S_{i_2}\cup...\cup S_{i_t}\right)\\
		&& -L\left(S_{i_1}\cup S_{i_2}\cup...\cup S_{i_t}\right),
	\end{eqnarray*} 
	and we can compute $L\left(S_j\cup S_{i_1}\cup S_{i_2}\cup...\cup S_{i_t}\right)$ and $L\left(S_{i_1}\cup S_{i_2}\cup...\cup S_{i_t}\right)$ via Theorem~\ref{lattice-basic-thm} and Corollary~\ref{coro-lattice}, then at each stage, we just need to traversal $m$ polyhedra to find $S_j.$
	
	Let $V(l)$ be the number of lattice points of coverage for $l$ regions, and let $OPT$ be the maximum number of lattice points of union of a optimal solution $S_{i_1},\, S_{i_2},\,...,\,S_{i_k}$ for the polyhedra maximum lattice coverage problem. 
	
	First, we show that 
	$$V(t)\geq \left(1-\left(1-\frac{1}{k}\right)^t\right)OPT$$
	by using induction method.
	
	It is trivial at $t=1,$ since $V(1)\geq \frac{OPT}{k}.$ Assume $V(t)\geq \left(1-\left(1-\frac{1}{k}\right)^t\right)OPT.$ Consider the case $t+1.$
	\begin{eqnarray*}
		V(t+1)&=&V(t)+L\left(S_{i_{t+1}}- \left(S_{i_1}\cup S_{i_2}\cup...\cup S_{i_t}\right)\right)\\
		&\geq& V(t)+\frac{OPT-V(t)}{k}\\
		&=&\left(1-\frac{1}{k}\right)V(t)+\frac{OPT}{k}\\
		&\geq&\left(1-\left(1-\frac{1}{k}\right)^t\right)\left(1-\frac{1}{k}\right)OPT+\frac{OPT}{k}\\
		&=&\left(1-\left(1-\frac{1}{k}\right)^{t+1}\right)OPT.
	\end{eqnarray*} 
	
	Therefore,
	$$V(t)\geq \left(1-\left(1-\frac{1}{k}\right)^t\right)OPT.$$

	Since $\left(1-\frac{1}{k}\right)^k$ is increasing and $\frac{1}{e}=\left(1-\frac{1}{k}\right)^k+\Omega\left(\frac{1}{k}\right),$ therefore,
	$$V(k)\geq \left(1-\left(1-\frac{1}{k}\right)^k\right)OPT\geq \left(1-\frac{1}{e}\right)OPT.$$
	
	There are $k$ stages. At each stage, the algorithm needs to check $m$ polyhedra in order to select the next region that covers the maximum number of lattice points that are uncovered. By Definition~\ref{def20}, $T_{U,\,L}(d,\,n)$ is the running time to compute the number of lattice points of the union of polyhedra, therefore the running time of the algorithm is $\bigO\left(km\cdot T_{U,\,L}(d,\,n)\right).$
	
\end{proof}

\section{Application in Polyhedra $\left(1-\beta\right)$-Lattice Set Cover Problem}\label{set_cover_sec}

In this section, we show how to apply the method at
Section~\ref{basic-sec} to the polyhedra $\left(1-\beta\right)$-lattice set cover problem. Before present the algorithm, we give some definitions about polyhedra $\left(1-\beta\right)$-lattice set cover problem.

\begin{definition} {\it $\left(1-\beta\right)$-Lattice Set Cover Problem:}
	For a real $\beta\in [0,\, 1),$ and a set of regions $S_1,\, S_2,\,\cdots,\,
	S_m$, select $k$ regions $S_{i_1},\,\cdots,\, S_{i_k}$ such that
	$L(S_{i_1}\cup S_{i_2}\cup\cdots\cup S_{i_k})\ge (1-\beta)L(S_1\cup \cdots \cup S_m)$ with $L(P)$ denotes the number of lattice points of region $P$.
\end{definition}

\begin{definition} {\it Polyhedra $\left(1-\beta\right)$-Lattice Set Cover Problem:}
	It is a lattice set cover problem when $S_1,\,\cdots,\, S_m$ are all
	polyhedra.
\end{definition}


\begin{theorem}
	For reals $\alpha\in (0, 1)$ and $\beta\in [0,\, 1),$ there is a $\bigO\left(m^2\cdot T_{U,\,L}(d,\,n)\right)$ time approximation algorithm for the polyhedra
	$\left(1-\beta\right)$-lattice set cover problem with the output $k$ regions $S_{i_1},\,...,\,S_{i_k}$ satisfying  $L\left(S_{i_1}\cup\cdots\cup S_{i_k}\right)\geq \left(1-\alpha\right)\left(1-\beta\right)L\left(S_{1}\cup\cdots\cup S_{m}\right)$ and $k\le \left(1+\ln\frac{1}{\alpha}\right)H,$
	where $H$ is the number of sets in an optimal solution for the polyhedra
	$\left(1-\beta\right)$-lattice set cover problem, and $T_{U,\,L}(d,\,n)$ is the running time to count the number of lattice points in the union of polyhedra with $n$ be the total number of linear inequalities from input polyhedra.
\end{theorem}

\begin{proof}
	It follows from greedy algorithm and Theorem~\ref{volume-basic-thm}. Let $L(P)$ denotes the number of lattice points of region $P$ and let $V$ denotes $L\left(S_{1}\cup\cdots \cup S_{m}\right).$ The greedy method works by picking, at each stage, the algorithm select region S that covers the maximum number of lattice points that are uncovered.
	Assume that $t$ regions have been selected, say, $S_{i_1},\, S_{i_2},\,...,\,S_{i_t},$ the greedy approach needs to select $S_j$ such that $L\left(S_j- \left(S_{i_1}\cup S_{i_2}\cup...\cup S_{i_t}\right)\right)$ is maximal for the next choice. The algorithm stops if $L\left(S_{i_1}\cup\cdots\cup S_{i_k}\right)\geq \left(1-\alpha\right)\left(1-\beta\right)L\left(S_{1}\cup\cdots\cup S_{m}\right).$
	
	Since 
	\begin{eqnarray*}
		L\left(S_j- \left(S_{i_1}\cup S_{i_2}\cup...\cup S_{i_t}\right)\right)&=&L\left(S_j\cup S_{i_1}\cup S_{i_2}\cup...\cup S_{i_t}\right)\\
		&& -L\left(S_{i_1}\cup S_{i_2}\cup...\cup S_{i_t}\right),
	\end{eqnarray*} 
	and we can compute $L\left(S_j\cup S_{i_1}\cup S_{i_2}\cup...\cup S_{i_t}\right)$ and $L\left(S_{i_1}\cup S_{i_2}\cup...\cup S_{i_t}\right)$ via Theorem~\ref{lattice-basic-thm} and Corollary~\ref{coro-lattice}, then at each stage, we just need to traversal $m$ polyhedra to find $S_j.$
	
	First, we show that 
	$$L\left(S_{i_1}\cup\cdots\cup S_{i_t}\right)\geq \left(1-\left(1-\frac{1}{H}\right)^{t}\right)\left(1-\beta\right)V$$
	by using induction method.
	
	It is trivial at $t=1,$ since $L\left(S_{i_1}\right)\geq \frac{\left(1-\beta\right)V}{H}=\left(1-\left(1-\frac{1}{H}\right)\right)\left(1-\beta\right)V.$ Assume $L\left(S_{i_1}\cup\cdots\cup S_{i_t}\right)\geq \left(1-\left(1-\frac{1}{H}\right)^{t}\right)\left(1-\beta\right)V.$ Consider the case $t+1.$
	\begin{eqnarray*}
		L\left(S_{i_1}\cup\cdots S_{i_t}\cup S_{i_{t+1}}\right)&=&L\left(S_{i_1}\cup\cdots\cup S_{i_t}\right)\\
		&&+L\left(S_{i_{t+1}}- \left(S_{i_1}\cup S_{i_2}\cup...\cup S_{i_t}\right)\right)\\
		&\geq& L\left(S_{i_1}\cup\cdots\cup S_{i_t}\right)+\frac{\left(1-\beta\right)V-L\left(S_{i_1}\cup\cdots\cup S_{i_t}\right)}{H}\\
		&=&\left(1-\frac{1}{H}\right)L\left(S_{i_1}\cup\cdots\cup S_{i_t}\right)+\frac{\left(1-\beta\right)V}{H}\\
		&\geq&\left(1-\left(1-\frac{1}{H}\right)^t\right)\left(1-\frac{1}{H}\right)\left(1-\beta\right)V+\frac{\left(1-\beta\right)V}{H}\\
		&=&\left(1-\left(1-\frac{1}{H}\right)^{t+1}\right)\left(1-\beta\right)V.
	\end{eqnarray*} 
	
	Therefore,
	\begin{equation}\label{eq22}
	L\left(S_{i_1}\cup\cdots\cup S_{i_t}\right)\geq \left(1-\left(1-\frac{1}{H}\right)^{t}\right)\left(1-\beta\right)V.
	\end{equation}
	
	For each $k$ satisfying 
	\begin{eqnarray}\label{eq33}
	&&\left(1-\frac{1}{H}\right)^k\leq\alpha
	\end{eqnarray} 
	then it will satisfy inequality~(\ref{eq22}).
	
	Since 
	\begin{eqnarray*}
		&&\left(1-\frac{1}{H}\right)^k=\left(1-\frac{1}{H}\right)^{H\cdot \frac{k}{H}}<\left(\frac{1}{e}\right)^{\frac{k}{H}},
	\end{eqnarray*} 
	then we have $k$ as 
	$$k\geq H\cdot \ln\frac{1}{\alpha}$$
	that will satisfy the inequality~(\ref{eq33}), which means it will also satisfy inequality~(\ref{eq22}). We can let
	$$k= \ceiling{H\cdot\ln \frac{1}{\alpha}}\le H\cdot\ln \frac{1}{\alpha}+1\le H\left(1+\ln \frac{1}{\alpha}\right).$$
	
	Therefore, 
	$$k\le \left(1+\ln\frac{1}{\alpha}\right)H.$$
	
	The algorithm has to traversal $m$ polyhedra. At each stage, the algorithm needs to check $m$ polyhedra in order to select the next region that covers the maximum number of lattice points that are uncovered. By Definition~\ref{def20}, the running time to compute the number of lattice points of the union of polyhedra is $\bigO\left(T_{U,\,L}(d,\,n)\right),$ therefore, the running time of the algorithm is $\bigO\left(m^2\cdot T_{U,\,L}(d,\,n)\right).$
\end{proof}

\section{Application in $\left(1-\beta\right)$-Continuous Polyhedra Set Cover}\label{conti-set-cover-app}
In this section, we show how to apply the method at
Section~\ref{basic-sec} to $\left(1-\beta\right)$-continuous polyhedra set cover problem. Before present the algorithm, we give some definitions about $\left(1-\beta\right)$-continuous polyhedra set cover problem.

\begin{definition} {\it $\left(1-\beta\right)$-Continuous Set Cover Problem:}
	For a real $\beta\in [0,\, 1),$ and a set of regions $S_1,\,
	S_2,\,\cdots,\, S_m$, select $k$ regions $S_{i_1},\cdots, S_{i_k}$ such
	that $vol(S_{i_1}\cup S_{i_2}\cup\cdots\cup S_{i_k})\ge (1-\beta)vol(S_1\cup \cdots \cup S_m)$ with $vol(P)$ denotes the volume of region $P$.
\end{definition}

\begin{definition} {\it $\left(1-\beta\right)$-Continuous Polyhedra Set Cover Problem:}
	It is a $\left(1-\beta\right)$-continuous set cover problem when $ S_1,\,\cdots,\, S_m$ are all polyhedra.
\end{definition}

\begin{theorem}
	For reals $\alpha\in (0, 1)$ and $\beta\in [0,\, 1),$ there is a $\bigO\left(m^2\cdot T_{U,\,V}(d,\,n)\right)$ time approximation algorithm for the
	$\left(1-\beta\right)$-continuous polyhedra set cover problem with the output $k$ regions $S_{i_1},\,...,\,S_{i_k}$ satisfying  $vol\left(S_{i_1}\cup\cdots\cup S_{i_k}\right)\geq \left(1-\alpha\right)\left(1-\beta\right)vol\left(S_{1}\cup\cdots\cup S_{m}\right)$ and $k\le \left(1+\ln\frac{1}{\alpha}\right)H,$ where $H$ is the number of sets in an optimal solution for the $(1-\beta)$-continuous set cover problem, and $T_{U,\,V}(d,\,n)$ is the running time to compute the volume of the union of polyhedra with $n$ be the total number of linear inequalities from input polyhedra.
\end{theorem}

\begin{proof}
	It follows from greedy algorithm and Theorem~\ref{volume-basic-thm}. Let $vol(P)$ denotes the volume of region $P$ and let $V$ denotes $vol\left(S_{1}\cup\cdots \cup S_{m}\right).$ The greedy method works by picking, at each stage, the algorithm select region S that covers the maximum volume that are uncovered.
	Assume that $t$ regions have been selected, say, $S_{i_1},\, S_{i_2},\,...,\,S_{i_t},$ the greedy approach needs to select $S_j$ such that $vol\left(S_j- \left(S_{i_1}\cup S_{i_2}\cup...\cup S_{i_t}\right)\right)$ is maximal for the next choice. The algorithm stops if $vol\left(S_{i_1}\cup\cdots\cup S_{i_k}\right)\geq \left(1-\alpha\right)\left(1-\beta\right)vol\left(S_{1}\cup\cdots\cup S_{m}\right).$
	
	Since 
	\begin{eqnarray*}
		vol\left(S_j- \left(S_{i_1}\cup S_{i_2}\cup...\cup S_{i_t}\right)\right)&=&vol\left(S_j\cup S_{i_1}\cup S_{i_2}\cup...\cup S_{i_t}\right)\\
		&& -vol\left(S_{i_1}\cup S_{i_2}\cup...\cup S_{i_t}\right),
	\end{eqnarray*} 
	and we can compute $vol\left(S_j\cup S_{i_1}\cup S_{i_2}\cup...\cup S_{i_t}\right)$ and $vol\left(S_{i_1}\cup S_{i_2}\cup...\cup S_{i_t}\right)$ via Theorem~\ref{volume-basic-thm} and Corollary~\ref{coro-volum}, then at each stage, we just need to traversal $m$ polyhedra to find $S_j.$
	
	First, we show that 
	$$vol\left(S_{i_1}\cup\cdots \cup S_{i_t}\right)\geq \left(1-\left(1-\frac{1}{H}\right)^{t}\right)\left(1-\beta\right)V$$
	by using induction method.
	
	It is trivial at $t=1,$ since $vol\left(S_{i_1}\right)\geq \frac{\left(1-\beta\right)V}{H}=\left(1-\left(1-\frac{1}{H}\right)\right)\left(1-\beta\right)V.$ Assume $vol\left(S_{i_1}\cup\cdots\cup S_{i_t}\right)\geq \left(1-\left(1-\frac{1}{H}\right)^{t}\right)\left(1-\beta\right)V.$ Consider the case $t+1.$
	\begin{eqnarray*}
		vol\left(S_{i_1}\cup\cdots S_{i_t}\cup S_{i_{t+1}}\right)&=&vol\left(S_{i_1}\cup\cdots \cup S_{i_t}\right)\\
		&&+vol\left(S_{i_{t+1}}- \left(S_{i_1}\cup S_{i_2}\cup...\cup S_{i_t}\right)\right)\\
		&\geq& vol\left(S_{i_1}\cup\cdots\cup S_{i_t}\right)+\frac{\left(1-\beta\right)V-vol\left(S_{i_1}\cup\cdots\cup S_{i_t}\right)}{H}\\
		&=&\left(1-\frac{1}{H}\right)vol\left(S_{i_1}\cup\cdots\cup S_{i_t}\right)+\frac{\left(1-\beta\right)V}{H}\\
		&\geq&\left(1-\left(1-\frac{1}{H}\right)^t\right)\left(1-\frac{1}{H}\right)\left(1-\beta\right)V+\frac{\left(1-\beta\right)V}{H}\\
		&=&\left(1-\left(1-\frac{1}{H}\right)^{t+1}\right)\left(1-\beta\right)V.
	\end{eqnarray*} 
	
	Therefore,
	\begin{equation}\label{eq2}
		vol\left(S_{i_1}\cup\cdots\cup S_{i_t}\right)\geq \left(1-\left(1-\frac{1}{H}\right)^{t}\right)\left(1-\beta\right)V.
	\end{equation}
	
	For each $k$ satisfying 
	\begin{eqnarray}\label{eq3}
		&&\left(1-\frac{1}{H}\right)^k\leq\alpha
	\end{eqnarray} 
	then it will satisfy inequality~(\ref{eq2}).
	
	Since 
	\begin{eqnarray*}
		&&\left(1-\frac{1}{H}\right)^k=\left(1-\frac{1}{H}\right)^{H\cdot \frac{k}{H}}<\left(\frac{1}{e}\right)^{\frac{k}{H}},
	\end{eqnarray*} 
	then we have $k$ as 
	$$k\geq H\cdot \ln\frac{1}{\alpha}$$
	that will satisfy the inequality~(\ref{eq3}), which means it will also satisfy inequality~(\ref{eq2}). We can let
	$$k= \ceiling{H\cdot\ln \frac{1}{\alpha}}\le H\cdot\ln \frac{1}{\alpha}+1\le H\left(1+\ln \frac{1}{\alpha}\right).$$
	
	Therefore, 
	$$k\le \left(1+\ln\frac{1}{\alpha}\right)H.$$
	
	The algorithm has to traversal $m$ polyhedra. At each stage, the algorithm needs to check $m$ polyhedra in order to select the next region that covers the maximum volume that are uncovered. By Definition~\ref{def20}, the running time to compute the volume of the union of polyhedra is $\bigO\left(T_{U,\,V}(d,\,n)\right),$ therefore, the running time of the algorithm is $\bigO\left(m^2\cdot T_{U,\,V}(d,\,n)\right).$

\end{proof}

\section{NP-hardness and Inapproximation}\label{NP-hardness}
In this section, we show the NP-hardness of the maximum coverage
problem and set cover problem when each set is represented as union
of polyhedrons.

\begin{definition}The maximum region coverage problem is that given 
a list  $A_1,\cdots, A_m$ of regions in $\mathbb{R}^d$, and an integer $k\ge 1$. Find $k$ regions with the largest volume for their union.
\end{definition}

\begin{definition}
	Let $A_1,\cdots, A_k$ be a list axis parallel rectangles. Let $A$ be
	the region to be union  $A_1\cup\cdots\cup A_k$. Then $A$ is called
	an {\it axis parallel rectangle union region}.
\end{definition}

\begin{definition}
The maximum region coverage problem with each set to be axis
parallel rectangle union region is called {\it axis parallel
rectangle union  maximum region coverage problem}.
\end{definition}

\begin{definition}
The lattice set cover  problem with each set to be all lattice points in  axis parallel
rectangle union region is called {\it lattice axis parallel
rectangle union region maximum coverage problem}.
\end{definition}

\begin{theorem}\label{inapp-thm}The following statements are true:
\begin{enumerate}
\item
The axis parallel rectangle union  maximum region coverage problem is
NP-hard. Furthermore, if there is a polynomial time
$c$-approximation for axis parallel rectangle union maximum region
coverage problem, then there is a polynomial time $c$-approximation
for classical maximum coverage problem.
\item
The axis parallel rectangle union set problem is NP-hard.
Furthermore, if there is a polynomial time $c$-approximation for
lattice axis parallel rectangle union region set cover problem, then
there is a polynomial time $c$-approximation for classical set cover
problem.
\end{enumerate}
\end{theorem}

\begin{proof}
Let $S_1,\cdots, S_m$ be an input for a (classical) maximum coverage problem or set cover problem.
Let $U$ be $S_1\cup S_2\cup\cdots\cup S_m$. We construct a axis
parallel rectangle union  maximum region coverage problem. Without
loss of generality, we assume $U=\{1,2,\cdots, n\}$. For each $S_i$,
we construct $A_i$, which is an axis parallel rectangle region. Let
$D$ represent the line segment of length $2n$ from $(0, {1\over 2})$
to $(2n-1,{1\over 2})$. Let $B_k$ be the $1\times 1$ square with
left bottom corner at $(2(i-1), 0)$. Let $A_i=D\cup (\cup_{j\in S_i}
B_i)$. It is easy to see that if $|S_{i_1}\cup \cdots \cup
S_{i_k}|=h$ if and only if the area of $A_{i_1}\cup \cdots\cup
A_{i_k}$ is $h$.

If $S_{i_1}\cup \cdots \cup S_{i_k}=S_1\cup \cdots \cup S_m$ if and
only if the area of $A_{i_1}\cup \cdots\cup A_{i_k}$ is $n$.
It is easy to see to this reduction reserves the ratio of approximation for both maximum coverage and set cover problems.
\end{proof}

Using reasonable hypothesis of complexity theory, 
Feige~\cite{Feige98} showed that $1-{1\over e}$ is the best possible polynomial time approximation ratio for the maximum coverage problem, and $\ln n$ is the best possible polynomial time approximation ratio for the set cover problem, where $n$ is the number of elements in the universe set for the classical set cover. Our Theorem~\ref{inapp-thm} shows the connection of inapproximation between the classical problems and their corresponding geometric versions.

\section{Conclusions}\label{concl}
We introduce polyhedra circuits. Each polyhedra circuit characterizes a geometric region in $\mathbb{R}^d$. They can be applied to represent a rich class of geometric objects, which include all polyhedra and the union of a finite number of polyhedra.  They can be used to approximate a large class of $d$-dimensional manifolds in $\mathbb{R}^d$. Define $T_V(d,\, n)$ be the polynomial time in $n$ to compute the volume of one rational polyhedra, $T_L(d,\, n)$ be the polynomial time in $n$ to count the number of lattice points in one rational polyhedra with $d$ be a fixed dimensional number, $T_I(d,\, n)$ be the polynomial time in $n$ to solve integer linear programming time when the fixed dimensional number $d$ is small, where $n$ is the total number of linear inequalities from input polyhedra. We develop algorithms to count the number of lattice points in the geometric region determined by a polyhedra circuit in $\bigO\left(nd\cdot r_d(n)\cdot T_V(d,\, n)\right)$ time and to compute the volume of the geometric region determined by a polyhedra circuit in $\bigO\left(n\cdot r_d(n)\cdot T_I(d,\, n)+r_d(n)T_L(d,\, n)\right)$ time, where $n$ is the number of input linear inequalities, $d$ is number of variables and $r_d(n)$ be the maximal number of regions that $n$ linear inequalities with $d$ variables partition $\mathbb{R}^d$.Define $T_{U,\,L}(d,\,n)$ be the running time to count the number of lattice points in the union of polyhedra and $T_{U,\,V}(d,\,n)$ be the running time to compute the volume of the union of polyhedra when given a list of polyhedra where $n$ is the total number of linear inequalities from input polyhedra and $d$ is number of varibales.
We applied the algorithms to continuous polyhedra maximum coverage problem and polyhedra maximum lattice coverage problem. The algorithm for continuous polyhedra maximum coverage problem gives $\left(1-\frac{1}{e}\right)$-approximation in $\bigO\left(km\cdot T_{U,\,V}(d,\,n)\right)$ running time and the algorithm for polyhedra maximum lattice coverage problem gives $\left(1-\frac{1}{e}\right)$-approximation in $\bigO\left(km\cdot T_{U,\,L}(d,\,n)\right)$ running time.
We also applied the algorithms to polyhedra $\left(1-\beta\right)$-lattice set cover problem and $\left(1-\beta\right)$-continuous polyhedra set cover problem. The algorithm for polyhedra $\left(1-\beta\right)$-lattice set cover problem
such that the optimal $k$ regions cover $(1-\alpha)(1-\beta)$ of the number of lattice points $V$ of region $S_1\cup\cdots \cup S_m$ with $k\le \left(1+\ln\frac{1}{\alpha}\right)H$ in $\bigO\left(m^2\cdot T_{U,\,L}(d,\,n)\right)$ time,
where $H$ is the number of sets in an optimal solution for the polyhedra
$\left(1-\beta\right)$-lattice set cover problem, and 
the algorithm for $\left(1-\beta\right)$-continuous polyhedra set cover problem such that the optimal $k$ regions cover $(1-\alpha)(1-\beta)$ of the volume $V$ of region $S_1\cup\cdots \cup S_m$ with $k\le \left(1+\ln\frac{1}{\alpha}\right)H$ in $\bigO\left(m^2\cdot T_{U,\,V}(d,\,n)\right)$ running time,
where $H$ is the number of sets in an optimal solution for the $(1-\beta)$-continuous set cover problem. We also showed the NP-hardness of the continous maximum coverage problem and set cover problem when each set is represented as union of polyhedra.







\end{document}